\documentclass[11pt]{article}

\usepackage[margin=0.8in]{geometry}

\usepackage{varioref}
\usepackage[usenames,svgnames,xcdraw,table]{xcolor}
\definecolor{DarkBlue}{rgb}{0.1,0.1,0.5}
\definecolor{DarkGreen}{rgb}{0.1,0.5,0.1}

\usepackage{nicefrac}

\usepackage{hyperref}
\hypersetup{
   colorlinks   = true,
 linkcolor    = DarkBlue, 
 urlcolor     = DarkBlue, 
	 citecolor    = DarkGreen 
}

\usepackage{appendix}

\usepackage{mathtools}

\usepackage{algorithmic}
\usepackage{algorithm}
\usepackage{array}

\usepackage[font={small,it}]{caption}

\usepackage{graphicx,epsfig,amsmath,latexsym,amssymb,verbatim}

\usepackage[square,semicolon,numbers]{natbib}

\newcommand{\extra}[1]{}

\usepackage{amsmath,amssymb,amsfonts}
\usepackage{amsthm}
\usepackage{thmtools,thm-restate}
\usepackage{enumitem}

\newtheorem{theorem}{Theorem}

\newtheorem{definition}{Definition}
\newtheorem{lemma}{Lemma}

\def\squareforqed{\hbox{\rlap{$\sqcap$}$\sqcup$}}
\def\qed{\ifmmode\squareforqed\else{\unskip\nobreak\hfil
\penalty50\hskip1em\null\nobreak\hfil\squareforqed
\parfillskip=0pt\finalhyphendemerits=0\endgraf}\fi}
\def\endenv{\ifmmode\;\else{\unskip\nobreak\hfil
\penalty50\hskip1em\null\nobreak\hfil\;
\parfillskip=0pt\finalhyphendemerits=0\endgraf}\fi}

\renewenvironment{proof}{\noindent \textbf{{Proof~} }}{\qed\medskip}
\newenvironment{proof+}[1]{\noindent \textbf{{Proof #1~} }}{\qed\medskip}

\mathchardef\ordinarycolon\mathcode`\:
\mathcode`\:=\string"8000
\def\vcentcolon{\mathrel{\mathop\ordinarycolon}}
\begingroup \catcode`\:=\active
  \lowercase{\endgroup
  \let :\vcentcolon
  }

\usepackage{amsmath}
\usepackage{amssymb,amsthm,mathtools}

\usepackage{amsmath,amsfonts} 
\usepackage{caption} 
\usepackage{xcolor}
\usepackage{hyperref}
\usepackage{comment}

\usepackage{verbatim}

\usepackage[font=small,labelfont=bf]{caption}
\usepackage{xspace}
\usepackage{etoolbox}

\title{\bfseries Existence and Computation of Maximin Fair Allocations \\ Under Matroid-Rank Valuations}

\author{Siddharth Barman\thanks{Indian Institute of Science. {\tt barman@iisc.ac.in}} \and Paritosh Verma\thanks{Indian Institute of Science. {\tt paritoshverma97@gmail.com}}}

\date{}

\newcommand{\EFone}{\textsc{EF}1}
\newcommand{\EFx}{\textsc{EFx}}

\newcommand{\MMS}{\textsc{MMS}}
\newcommand{\PMMS}{\textsc{PMMS}}
\newcommand{\alloc}{\mathcal{A}}
\newcommand{\palloc}{\mathcal{P}}
\newcommand{\SW}{\operatorname{SW}}
\newcommand{\umat}{\widehat{\mathcal{M}}}
\newcommand{\urank}{\widehat{r}}

\begin{document}

\maketitle 

\begin{abstract}
We study fair and economically efficient allocation of indivisible goods among agents whose valuations are rank functions of matroids. Such valuations constitute a well-studied class of submodular functions (i.e., they exhibit a diminishing returns property) and model preferences in several resource-allocation settings. We prove that, for matroid-rank valuations, a social welfare-maximizing allocation that gives each agent her maximin share always exists. Furthermore, such an allocation can be computed in polynomial time. We establish similar existential and algorithmic results for the pairwise maximin share guarantee as well.

To complement these results, we show that if the agents have  binary XOS valuations or weighted-rank valuations, then maximin fair allocations are not guaranteed to exist. Both of these valuation classes are immediate generalizations of matroid-rank functions. 
\end{abstract}

\section{Introduction}
Discrete fair division is an active field of work at the interface of computer science and mathematical economics. This area studies fair and economically efficient allocation of goods (resources) that cannot be fractionally assigned. Indeed, many real-world settings (such as course allocation \cite{budish2017course} and division of inheritance) entail assignment of discrete resources. Motivated, in part, by such domains, a significant body of work in recent years has been directed towards notions of fairness (and complementary algorithms) that are applicable in the indivisible context; see, e.g.,~\cite{BCEL+16} and~\cite{endriss2017trends} for textbook expositions. Arguably, the two most prominent fairness notions in discrete fair division are envy-freeness up to one good~\cite{Bud11,LMMS042} and the maximin share guarantee~\cite{Bud11}. 

An allocation (a partition) of the indivisible goods among the agents is said to be envy-free up to one good ($\EFone$) iff each agent values her own bundle over the bundle of any other agent, up to the removal of some good from the other agent's bundle. $\EFone$ provides a cogent relaxation of the classic fairness criterion of envy freeness\footnote{An allocation is said to be envy-free iff every agent values her bundle at least as much as she values any other agent's bundle.} in the indivisible goods context. While simple examples rule out the existence of envy-free allocations of indivisible goods, $\EFone$ allocations are guaranteed to exist under mild assumptions. In particular, if the agents valuations are monotonic, then an $\EFone$ allocation exists and can be computed in polynomial time \cite{LMMS042}.

Another relaxation of envy-freeness is obtained via maximin shares, that---for each agent $i$---is defined to be the maximum value that $i$ can achieve by partitioning all the goods into $n$ bundles and then receiving a minimum valued (according to $i$) one; throughout, $n$ will denote the number of agents participating in the fair division exercise. An allocation is said to be a maximin share ($\MMS$) allocation iff each agent receives a bundle of value at least as much as her maximin share. That is, maximin shares correspond to an intuitive threshold and, an allocation is deemed to be fair, under this criterion, iff the agent-specific threshold is met for every agent. Maximin share allocations are not guaranteed to exist--this result holds, in particular, for instances in which the agents' valuations are additive \cite{KPW16,PW14}.\footnote{An agent's valuation $v$ is said to be additive iff, for any subset of goods $S$, we have $v(S) = \sum_{g \in S} v(g)$; here $v(g)$ is the value that the agent has for good $g$.} However, this notion is quite amenable to approximation guarantees: under additive valuations, an allocation that assigns each agent a bundle of value at least $\left(\frac{3}{4} + \frac{1}{12n} \right)$ times her maximin share exists~\cite{garg2020improved,ghodsi2018fair}. In addition, for submodular valuations, a $1/3$-approximate maximin share allocation exists and can be found in polynomial time \cite{ghodsi2018fair}; recall that a set function $v$ is said to be submodular iff it satisfies the following diminishing returns property: $v(A \cup \{g \}) - v(A) \geq v(B \cup \{g\}) - v(B)$ for all subsets $A \subseteq B$ and $g \notin B$.   

With relevant fairness criteria (such as $\EFone$ and $\MMS$) in hand, central threads of research in discrete fair division are aimed at understanding (i) the existence of fairness notions, (ii) their computational tractability, and (iii) the impact of fairness guarantees on economic efficiency. The current work contributes to these key themes by establishing positive results for the maximin share guarantee in the context of matroid-rank valuations. 

Rank functions of matroids provide a combinatorial generalization of linear-algebraic notions of independence and rank. They constitute a well-studied class of submodular functions and, in fact, admit the following characterization~\cite[Chapter~39]{schrijver2003combinatorial}: every submodular function $r$ with binary marginals is a matroid-rank function. Recall that a set function $r$ is said to have binary marginals iff $r(A \cup \{ g \}) - r(A) \in \{0,1\}$, for all subsets  $A$ and elements $g$. Rank functions model preferences in several resource-allocation settings. For instance, in the fair allocation of public housing units, the underlying preferences can be expressed as matroid-rank functions~\cite{deng2013story,benabbou2020finding}: to achieve fairness across different ethnic groups, one can model each group as an agent and for a subset $S$ of housing units the utility of an agent/group is obtained by matching the group members to the units in $S$. The matching is based on the members' binary preferences and the resulting (matching-based) valuation is a rank function of a (transversal) matroid. Benabbou et al.~\cite{benabbou2020finding} identify other domains wherein rank functions are applicable. 

For this relevant function class, the recent results of Babaioff et al.~\cite{babaioff2020fair} and Benabbou et al.~\cite{benabbou2020finding} develop polynomial-time algorithms for finding allocations that are both $\EFone$ and Pareto efficient.\footnote{Babaioff et al.~\cite{babaioff2020fair} additionally achieve truthfulness with \emph{Lorenz domination} as a fairness criterion, which implies $\EFone$, along with other fairness notions. They also note that the allocations computed by their mechanism are not guaranteed to be $\MMS$~\cite[Proposition~6]{babaioff2020fair}. Developing a truthful mechanism under rank functions for $\MMS$ remains an interesting direction of future work.}  The current work complements these $\EFone$ results by focussing on $\MMS$.

\paragraph{Our Contributions.} For matroid-rank valuations, we prove that a maximin share allocation is guaranteed to exist. Moreover, in this context, one can achieve fairness along with economic efficiency as well as computational tractability: under matroid-rank valuations, an $\MMS$ allocation that also maximizes social welfare (across all allocations) can be computed in polynomial time (Theorem~\ref{theorem:mms}). Note that a welfare-maximizing allocation is Pareto efficient as well. Also, recall that prior work on submodular valuations implies the existence of an allocation in which each agent receives a bundle of value at least $1/3$ times her maximin share \cite{ghodsi2018fair} and, specifically for matroid-rank functions, the result of Babaioff et al.~\cite{babaioff2020fair} shows that a $1/2$-approximate $\MMS$ allocation always exists. Hence, our existential result for exact $\MMS$ is a novel contribution. The proof of this result is constructive--we develop a polynomial-time algorithm (Algorithm \ref{algo:mms}) that is guaranteed to find the desired allocation. 

The algorithm for finding the desired allocation starts with a social welfare-maximizing allocation, which can be computed efficiently for matroid-rank valuations. Then, it iteratively performs local updates---by swapping goods among selected chains of bundles---till an $\MMS$ allocation is obtained. Even though, at a high level, the algorithm is direct, its analysis relies on interesting applications of deep results from matroid theory, e.g., the matroid union theorem (Lemma \ref{lemma:matroid-union-theorem}). The technical results in the current work highlight interesting connections between maximin shares and the rich literature of matroid theory.  

The maximin share of an agent can be conceptually interpreted through a discrete execution of the cut-and-choose protocol: agent $i$ partitions the goods and the other ($n-1$) agents get to pick a bundle before $i$; here, agent $i$---by maximizing over all $n$-partitions---can guarantee for herself a value equal to her maximin share, irrespective of the choices of the other agents. Building upon this interpretation, Caragiannis et al.~\cite{CKMP+19} consider a meaningful variant wherein the discrete cut-and-choose protocol is executed between all pairs of agents. Specifically, a collection of mutually disjoint bundles (of goods), $A_1, A_2, \ldots, A_n$, is said to satisfy the pairwise maximin share guarantee ($\PMMS$) iff for every pair of agents, $i$ and $j$, agent $i$'s value for her bundle, $A_i$, is at least as much as the share she would obtain by executing the discrete cut-and-choose protocol among two agents and goods in $A_i \cup A_j$. 

We also develop a polynomial-time algorithm (Algorithm \ref{algo:pmms}) for finding a partial allocation that satisfies the pairwise maximin share guarantee and maximizes social welfare, across all allocations (Theorem~\ref{theorem:pmms}). The computed allocation can be partial in the sense that it might not allocate all the goods. 

It is relevant to note that, under monotonic valuations, any partial allocation that satisfies the maximin share guarantee, can be extended into an $\MMS$ allocation which is also complete--one can simply include the unassigned goods into, say, the first agent's bundle; for this reason and ease of presentation, we do not explicitly distinguish between partial and complete allocations in the $\MMS$ context. Such an extension can, however, violate the $\PMMS$ guarantee. Indeed, the existence of a partial $\PMMS$ allocation does not---by itself---imply the existence of a complete $\PMMS$ allocation. While the existence of such a complete allocation remains an interesting open question, we prove that (for matroid-rank valuations) there always exists a $\PMMS$ allocation, which might not be complete, but it maximizes social welfare and, hence, is Pareto efficient (across all allocations, partial and complete). Therefore, this fairness guarantee is obtained without a loss in economic efficiency.

The work of Caragiannis et al.~\cite{CKMP+19} highlights the relevance of $\PMMS$ by showing that, under additive valuations, every $\PMMS$ allocation satisfies envy-freeness up to the removal of \emph{any} good ($\EFx$).\footnote{The universal existence of $\EFx$ allocations, under additive valuations, is an important open question in discrete fair division.} Since $\EFx$ is a stricter criterion than $\EFone$, we get that $\PMMS$ implies $\EFone$, under additive valuations. This implication continues to hold for matroid-rank functions; Theorem \ref{theorem:pmms-efone} (in Appendix \ref{appendix:pmms-efone}) shows that, under submodular valuations, every $\PMMS$ allocation is in fact $\EFone$. Using this implication and our polynomial-time algorithm for $\PMMS$, one can recover the result of Benabbou et al.~\cite{benabbou2020finding}, which shows that (under rank functions)  an $\EFone$ and social welfare-maximizing allocation can be computed in polynomial time. We note that the algorithm of Benabbou et al.~\cite{benabbou2020finding} and the mechanism of Babaioff et al.~\cite{babaioff2020fair} (both developed for matroid-rank valuations) also do not necessarily assign all the goods. Furthermore, a few recent results show that  keeping a subset of goods unassigned can be used to attain particular fairness guarantees along with a bounded loss in efficiency; see, e.g.,~\cite{chaudhury2020little,caragiannis2019envy}.

Notably, the converse implication from $\EFone$ to $\PMMS$ does not hold, i.e., there exist instances, with matroid-rank valuations, wherein particular $\EFone$ allocations do not satisfy the pairwise maximin share guarantee (Appendix~\ref{example:EFone-not-MMS}).

Finally, we show that---in contrast to the above-mentioned positive results---if the agents have  binary XOS valuations or weighted-rank valuations, then maximin fair allocations are not guaranteed to exist (Theorem~\ref{thm:xos} and~\ref{thm:weight-rank}). Both of these valuation classes are immediate generalizations of matroid-rank functions.

\paragraph{Additional Related Work.}
Binary additive valuations are a particular subclass of rank functions and have been studied in multiple fair-division results; see, e.g., \cite{BM04,BL16,KPS15,halpern2020fair}. These additive valuations model settings in which, for each agent, a good is either acceptable or not. In one of the initial results on maximin shares, Bouveret and Lema\^{i}tre \cite{BL16} showed that maximin share allocations exist under binary additive valuations: for such dichotomous valuations, any $\EFone$ allocation is $\MMS$ as well. Such an implication, however, does not hold with matroid-rank functions; e.g., Proposition 6 in ~\cite{babaioff2020fair} identifies instances in which all \emph{leximin}---and, hence, particular $\EFone$---allocations do not satisfy the maximin share guarantee. For completeness, in Appendix~\ref{example:EFone-not-MMS} we provide an example to show that, under matroid-rank valuations, $\EFone$ allocations are not guaranteed to be $\MMS$. 

Under binary additive valuations, $\EFone$ and $\PMMS$ are equivalent notions, i.e., an $\EFone$ allocation is guaranteed to be $\PMMS$ and vice versa. Hence, known algorithms for computing $\EFone$ allocations (e.g.,~\cite{LMMS042} ) are sufficient to compute $\PMMS$ allocations under binary additive valuations. However, these algorithms are not guaranteed to output a $\PMMS$ allocation for matroid-rank valuations, since $\EFone$ no longer implies $\PMMS$ (Appendix~\ref{example:EFone-not-MMS}). 

Few recent results in discrete fair division have considered additive valuations while utilizing matroids to express feasibility constraints; see, e.g., \cite{GMT14,GM17,BB18,DFS20}. These results are complementary to the present work on valuation functions based on matroids.

\section{Notation and Preliminaries}
\label{section:notation}
We consider partitioning $m$ indivisible goods among $n$ agents in a fair and economically efficient manner. We will, throughout, write $[m] \coloneqq \{1,2, \ldots, m\}$ to denote the set of goods and $[n] \coloneqq \{1,2,\ldots, n\}$ to denote the set of agents. The cardinal preference of each agent $i \in [n]$, over subsets of goods, is specified via a valuation function $v_i: 2^{[m]} \mapsto \mathbb{R}_+$; in particular, $v_i(S) \in \mathbb{R}_+$ denotes the value that agent $i \in [n]$ has for a subset of goods $S \subseteq [m]$. In this setup, an instance of the fair division problem is a tuple $\langle [m], [n], \{v_i\}_{i=1}^n \rangle$. For subsets $X \subseteq [m]$ and goods $g \in [m]$, we will use the shorthand $X + g \coloneqq X \cup \{g \}$ and $X - g \coloneqq X \setminus \{ g \}$. 

\paragraph{Allocations and Social Welfare.} For integer $k \in \mathbb{Z}_+$ and subset $S \subseteq [m]$, write $\Pi_k(S)$ to denote the set of all $k$-partitions of $S$. An allocation $\alloc = (A_1, A_2, \ldots, A_n) \in \Pi_n([m])$ is an $n$-partition of all the goods, i.e., $\cup_{i=1}^n A_i = [m]$ and $A_i \cap A_j = \emptyset$ for all $i \neq j$. Here each subset $A_i$ is assigned to agent $i \in [n]$ and will be referred to as a bundle.

The term \emph{partial allocation} will be used to denote a collection of pairwise-disjoint subsets of goods $\mathcal{P} = (P_1, P_2, \ldots, P_n)$, where $P_i$ is assigned to agent $i$. For a partial allocation $\palloc= (P_1, \ldots, P_n)$, the set of unallocated goods is $[m] \setminus \left( \bigcup_{i \in [n]} P_i \right)$. Note that, in contrast to an allocation, for a partial allocation $\mathcal{P} = (P_1, \ldots, P_n)$ it is not necessary that $\cup_{i=1}^n P_i = [m]$. Indeed, $\palloc$ is a complete allocation iff $[m] \setminus \left( \bigcup_{i \in [n]} P_i \right)= \emptyset$.

The social welfare $\SW(\cdot)$ of an (partial) allocation $\alloc=(A_1, \allowbreak \ldots, A_n)$ is the sum of the values that $\alloc$ generates among the agents, $\SW(\alloc) \coloneqq \sum_{i=1}^n v_i(A_i)$. 

\paragraph{Fairness Notions.} The notions of fairness considered in the current work are defined next. For a fair division instance $\langle [m], [n], \{v_i\}_{i=1}^n \rangle$, the \emph{maximin share} of an agent $i$ is defined as 
\[\mu_i \coloneqq \max_{(X_1, \ldots, X_n) \in \Pi_n([m])} \ \allowbreak \min_{j \in [n]}  v_i(X_j). \] 

We will also consider the following generalization of this quantity, with any number of agents $k \in \mathbb{Z}_+$ and subset of goods $S \subseteq [m]$ 
\begin{align}
\label{eq:mms-defn} 
\mu_i(k, S) & \coloneqq \max_{(Y_1, \ldots, Y_k) \in \Pi_k(S)} \ \min_{j \in [k]}  v_i(Y_j)
\end{align}
Note that $\mu_i = \mu_i(n, [m])$.
\begin{definition}[$\MMS$] An allocation $\alloc = (A_1, \ldots, A_n)$ is said to be a \emph{maximin share allocation} iff $v_i(A_i) \geq \mu_i$ for all agents $i \in [n]$.
\end{definition}

\begin{definition}[$\PMMS$] An allocation (partial or complete) $\palloc = (P_1, \ldots, P_n)$ is said to be a \emph{pairwise maximin share allocation} iff $v_i(P_i) \geq \mu_i(2, P_i \cup P_j)$ for all agents $i, j \in [n]$.
\end{definition}

As mentioned previously, the $\PMMS$ allocations computed by our algorithm might not be complete, but they maximize social welfare (across all allocations, partial and complete) and, hence, are Pareto efficient. Therefore, we obtain the $\PMMS$ guarantee without any loss in economic efficiency.

\begin{definition}[$\EFone$] An allocation (partial or complete) $\mathcal{A}$ is said to be envy-free up to one good ($\EFone$) iff for all agents $i, j \in [n]$, with $A_j \neq \emptyset$, there exists a good $g \in A_j$ such that $v_i(A_i) \geq v_i(A_j \setminus \{ g \})$. 
\end{definition}

While the current work primarily focuses on achieving $\MMS$ and $\PMMS$ guarantees exactly, one requires the following approximate versions of these notions when considering connections between different fairness concepts. For parameter $\alpha \in [0,1]$, an allocation $(A_1, \ldots, A_n)$ is said to be an $\alpha$-approximate maximin share ($\alpha$-$\MMS$) allocation iff $v_i(A_i) \geq \alpha \ \mu_i$ for all $i \in [n]$. Similarly, an $\alpha$-$\PMMS$ allocation (complete or partial) $(P_1, \ldots, P_n)$ is one that satisfies $v_i(P_i) \geq \alpha \ \mu_i(2, P_i \cup P_j)$ for all $i, j \in [n]$.

\paragraph{Matroids and rank functions.} The current work addresses settings in which the valuation of each agent $i \in [n]$ is a rank function of a matroid $\mathcal{M}_i = ([m], \mathcal{I}_i)$. Recall that a pair $([m], \mathcal{I})$ is called a matroid iff $\mathcal{I}$ is a nonempty collection of subsets of $[m]$ that satisfies (i) Hereditary property: if $I \in \mathcal{I}$ and $J \subseteq I$, then $J \in \mathcal{I}$, and (ii) Augmentation property: if $I, J \in \mathcal{I}$ and $|J| < |I|$, then there exists $g \in I \setminus J$ such that $J + g \in \mathcal{I}$. Given a matroid $\mathcal{M} = ([m], \mathcal{I})$, a subset $S \subseteq [m]$ is said to be \emph{independent} iff $S \in \mathcal{I}$. 

The rank function $r: 2^{[m]} \mapsto \mathbb{Z}_+$ of a matroid $\mathcal{M} = ([m], \mathcal{I})$ captures, for each subset $X \subseteq [m]$, the size of the largest (cardinality wise) independent subset within $X$; formally, 
\begin{align*}
r(X) & \coloneqq \max\{|I| \ : \  I \subseteq X \text{ and } I \in \mathcal{I}\}.
\end{align*}
Note that rank functions, by definition, are nonnegative ($r(X) \geq 0$ for all $X \subseteq [m]$) and monotone ($r(X) \leq r(Y)$ for all $X \subseteq Y$). Also, the following characterization is well known~\cite[Chapter~39]{schrijver2003combinatorial}: every submodular function $r$ with binary marginals\footnote{That is, $r(A + g) - r(A) \in \{0,1\}$, for all $A\subseteq [m]$ and $g \in [m]$.} is in fact a matroid-rank function. 

As mentioned previously, we will focus on fair division instances $\langle [m], [n], \{v_i\}_{i=1}^n \rangle$ in which, for each agent $i$, the valuation $v_i$ is the rank function of a matroid $\mathcal{M}_i = ( [m], \mathcal{I}_i )$. Hence, for any subset of goods $S \subseteq [m]$, we have $ v_i(S) \leq |S|$ and equality holds here iff $S$ is an independent set in $\mathcal{M}_i$, i.e., $S \in \mathcal{I}_i$. We will assume, throughout, that the valuations are specified via an oracle that answers value-queries: given any subset $S \subseteq [m]$ and an agent $i$, the oracle returns $v_i(S) \in \mathbb{R}_+$. That is, all of our algorithmic results hold in the basic value-oracle model and do not require an explicit description of the underlying matroids. 

\paragraph{Matroid union.} In matroid theory, the union operation enables one to construct a new matroid by combining independent sets of old ones. In particular, if $\mathcal{M}_1 = ([m], \mathcal{I}_1),\ldots, \mathcal{M}_n=([m], \mathcal{I}_n)$ are matroids, then their \emph{union} \[\umat \coloneqq \big( [m], \ \left\{ I_1 \cup \ldots \cup I_n : I_i \in \mathcal{I}_i \text{ for all } i \in [n] \right\} \big)\] is a matroid as well~\cite[Chapter~42]{schrijver2003combinatorial}. Write $\urank$ to denote the rank function of $\umat$. Note that a subset of goods $S\subseteq [m]$ is independent with respect to  $\umat$ (i.e., $\urank(S) =|S|$) iff $S$ admits an $n$-partition, $S_1, S_2, \ldots, S_n$, with the property that $S_i \in \mathcal{I}_i$, for all $i \in [n]$. Equivalently, a subset $S$ is independent in $\umat$ iff we can partition it among the agents and generate social welfare equal to $|S|$. {Recall that $S_i \in \mathcal{I}_i$ iff $v_i(S_i) = |S_i|$.} 

The next equation provides a direct connection between the rank function $\urank$ and the optimal social welfare in fair division instances $\langle [m], [n], \{v_i\}_{i=1}^n \rangle$ with matroid-rank valuations
\begin{align}
\label{eq:rank-sw}
\urank([m]) = \max_{\alloc \in \Pi_n([m])} \ \SW(\alloc)
\end{align}
To establish equation (\ref{eq:rank-sw}), note that given any (social-welfare maximizing) allocation $\alloc=(A_1, \ldots, A_n)$, for each $i$, there exists a subset $A'_i \subseteq A_i$ such that $v_i(A_i) = |A'_i|$ and $A'_i \in \mathcal{I}_i$; in particular, $A'_i$ is a maximum-size independent subset of $A_i$. For the partial allocation $\mathcal{A}' = (A'_1, \ldots, A'_n)$, we have $\SW(\alloc) = \SW(\alloc') = |\cup_{i=1}^n A'_i|$ and  $\cup_{i=1}^n A'_i$ is an independent set in $\umat$. Hence, the rank $\urank([m])$ is at least as much as the optimal social welfare. The reverse inequality follows from the fact that any (maximum-size) independent subset $P$ of $\umat$ admits a partition/partial allocation $\palloc = (P_1, \ldots, P_n)$ with the property that $\SW(\palloc) = \sum_{i=1}^n v_i(P_i) = \sum_{i=1}^n |P_i| = |P| = \urank(P)$.   

A classic result of Edmonds~\cite[Chapter~42.3]{schrijver2003combinatorial} shows that a maximum-size independent subset in $\umat$ (in particular, of size $\urank([m])$) can be computed in polynomial time. Using this \emph{matroid union algorithm}, we can efficiently find an allocation (partial or complete) $\alloc=(A_1, \ldots, A_n) \in \mathcal{I}_1 \times \ldots \times \mathcal{I}_n$ that maximizes social welfare in the given fair division instance. Note that, $\alloc$ might not be a complete allocation, but it maximizes social welfare among all allocations.

In the matroid union context, a central notion is that of an \emph{exchange graph}. Our results use this construct, along with the augmentation operation. For a partial allocation $\alloc = (A_1, \allowbreak \ldots, A_n) \in \mathcal{I}_1 \times \ldots \times \mathcal{I}_n$, comprised of independent sets, the \emph{exchange graph} $\mathcal{G}(\alloc)$ is a directed graph with vertex set as $[m]$ and it includes, for all $i \in [n]$, directed edge $(g, g') \in A_i \times \left([m] \setminus A_i\right)$ iff $A_i - g + g' \in \mathcal{I}_i$. That is, directed edge $(g, g') \in A_i \times \left([m] \setminus A_i\right)$ is included in $\mathcal{G}(\alloc)$ iff swapping along it maintains independence for $i$. 

Finally, we define \emph{path augmentation}: for a directed path $P =(g_1, g_2, \ldots, g_k)$ in the exchange graph $\mathcal{G}(\alloc)$ write
\begin{align*}
A_i \Delta P \coloneqq A_i \ \Delta \ \{ g_j , g_{j+1} : g_j \in A_i \} \qquad \text{for all $i \in [n]$.}
\end{align*}
That is, the set $A_i \Delta P$ is obtained by swapping along every directed edge $(g_j, g_{j+1})$ in $P$ that goes out of the set $A_i$. Recall that, for any two subsets $X$ and $Y$, the symmetric difference $X \Delta Y \coloneqq (X \setminus Y) \cup (Y \setminus X)$.

For any agent $i \in [n]$ and independent set $A_i \in \mathcal{I}_i$, we will denote by $F_i(A_i)$ the set of goods that can be included in $A_i$ while maintaining independence, $F_i(A_i) \coloneqq \{ g \in [m] \setminus A_i : A_i + g \in \mathcal{I}_i \}$. 

We will use the following known result  (stated in our notation). In particular, Lemma~\ref{lemma:path-augmentation} asserts that augmenting along shortest paths in the exchange graph maintains independence.\footnote{Following standard terminology, $P$ is said to be a shortest path between two vertex sets $V\subseteq [m]$ and $V' \subseteq [m]$ iff it has the fewest number of edges among all paths that connect any vertex in $V$ to any vertex in $V'$.} For completeness, we provide a proof of this result in Appendix~\ref{appendix:proof-of-path-augmentation}.

\begin{restatable}[\cite{schrijver2003combinatorial}]{lemma}{LemmaPathAugmentation}
\label{lemma:path-augmentation}
Let $\alloc = (A_1, \ldots, A_n) \in \mathcal{I}_1 \times \ldots \times \mathcal{I}_n$ be a (partial) allocation comprised of independent sets and, for agents $i \neq j$, let $P = (g_1, g_2, \ldots, g_t)$ be a shortest path in the exchange graph $\mathcal{G}(\alloc)$ between $F_i(A_i)$ and $A_j$ (in particular, $g_1 \in F_i(A_i)$ and $g_t \in A_j$). Then, for all $k \in [n] \setminus \{i, j\}$, we have $A_k \Delta P \in \mathcal{I}_k$ along with $\left( A_i \Delta P \right) + g_1 \in \mathcal{I}_i$ and $A_j - g_t \in \mathcal{I}_j$.
\end{restatable}

As a direct consequence of this lemma, we get that by augmenting along a shortest path, between $F_i(A_i)$ and $A_j$, one obtains a new (partial) allocation in which the valuation of agent $i$ increases by one and that of $j$ decreases by one. The valuations of all other agents remain unchanged. In particular, the social welfare remains unchanged after a path augmentation of the type identified in Lemma~\ref{lemma:path-augmentation}. 

The matroid union theorem \cite[Corollary~42.1a]{schrijver2003combinatorial} is stated next (in our notation). This result provides a convolution formula for the rank function $\urank$ of the union matroid $\umat$.

\begin{lemma}[\cite{schrijver2003combinatorial}]
\label{lemma:matroid-union-theorem}
Let $\umat$ be the union of matroids $\mathcal{M}_1, \ldots, \mathcal{M}_n$ with rank functions $v_1, \ldots, v_n$, respectively. Then, the rank function $\urank$ of $\umat$ satisfies 
\begin{align*}
\urank(S) = \min_{T \subseteq S} \ \left( |S \setminus T| + \sum_{i=1}^n v_i(T) \right) \qquad \text{ for all $S \subseteq [m]$.}
\end{align*}
\end{lemma}

\section{Maximin Share Guarantee for Matroid-Rank Valuations}
\label{section:mms}

This section develops a polynomial-time algorithm (Algorithm \ref{algo:mms}) for computing an $\MMS$ allocation that also maximizes social welfare. 

\floatname{algorithm}{Algorithm}
\begin{algorithm}[ht]
\caption{\textsc{AlgMMS}} \label{algo:mms}
\textbf{Input:} Fair division instance  $\langle [m],[n], \{v_i\}_i \rangle$ with value-oracle access to the matroid-rank valuations $v_i$s. \\
\textbf{Output:} Social welfare-maximizing and maximin share allocation $\mathcal{A} = (A_1, \ldots, A_n)$ 

  \begin{algorithmic}[1]
		\STATE Compute a social welfare maximizing (partial) allocation $\alloc = (A_1, \ldots, A_n) \in \mathcal{I}_1 \times \dots \times \mathcal{I}_n$   %
	    \STATE Initialize ${S}_{<} = \{i \in [n] \ : \ v_i(A_i) < \mu_i(n,[m])\}$ and  ${S}_{>} = \{i \in [n] \ : \ v_i(A_i) > \mu_i(n,[m])\}$ 
		\WHILE{${S}_< \neq \emptyset$} 
			\STATE Select any agent $i \in S_<$ and set $P =(g_1, \ldots, g_t)$ to be a shortest path from $F_i(A_i)$  to $ \bigcup_{j \in S_>} \ A_j$ \COMMENT{In particular, $g_1 \in F_i(A_i)$, with $i \in S_<$, and $g_t \in A_j$, with $j \in S_>$} \label{line:path}
			\STATE Update $A_k \gets A_k \Delta P$ for all $k \in [n] \setminus \{i,j\}$ \label{line:augment1}
			\STATE Update $A_i \gets (A_i \Delta P ) + g_1$ and $A_j \gets A_j - g_t$ \label{line:augment2}
			
			\STATE Set ${S}_{<} = \{i \in [n] \ : \ v_i(A_i) < \mu_i(n,[m])\}$ and  ${S}_{>} = \{i \in [n] \ : \ v_i(A_i) > \mu_i(n,[m])\}$
			\ENDWHILE
		\STATE \textbf{return } $\alloc^* = \left( A_1 \cup \left([m] \setminus \cup_{i=1}^n A_i \right), A_2, \ldots, A_n \right)$ \label{line:return}
		\end{algorithmic}
\end{algorithm}

Algorithm \ref{algo:mms} starts with a (partial) allocation $\alloc = (A_1, \ldots, A_n) \in \mathcal{I}_1 \times \dots \times \mathcal{I}_n$ that maximizes social welfare. As mentioned previously, a welfare-maximizing partial allocation (with independent $A_i$s) corresponds to a maximum-size independent set in $\umat$ (see equation \ref{eq:rank-sw}) and, hence, can be computed via the matroid union algorithm~\cite[Chapter~42.3]{schrijver2003combinatorial} using value-oracle access to the matroid rank-functions $v_i$s. Note that we are given oracle access to $v_i$s and not directly an oracle for the rank function of $\umat$.    

The matroid union algorithm also enables us to compute the maximin share, $\mu_i = \mu_i(n, [m])$, of each agent $i$ in polynomial time; see Appendix \ref{appendix:compute-mms}. With $\mu_i$s in hand, the algorithm iteratively updates $\alloc$---by path augmentation---till it becomes a maximin share allocation. Augmentation via shortest paths ensures that the social-welfare optimality of $\alloc$ and the independence of the bundles ($A_i \in \mathcal{I}_i$ for all $i$) is maintained throughout. 

At any point during the execution of Algorithm \ref{algo:mms}, with allocation $\alloc=(A_1,\ldots, A_n)$ in hand, we consider the set of agents whose current value is less than their maximin share, $S_< = \{i \in [n] \ : \ v_i(A_i) < \mu_i(n,[m])\}$ and the set of agents whose current value is more than their maximin share ${S}_{>} = \{i \in [n] \ : \ v_i(A_i) > \mu_i(n,[m])\}$. Augmenting along a shortest path $P$ in the exchange graph we increase the value of some agent $i \in S_<$ by one, at the cost of decreasing  the value of some agent $j \in S_>$ by one. The values of all other agents remain unchanged. We will show that, if $\alloc$ is not an $\MMS$ allocation, then such a path $P$ necessarily exists and after a polynomial number of iterations the while loop (that performs path augmentations) will terminate with an $\MMS$ allocation.

The following key lemma relies on an interesting invocation of the matroid union theorem (Lemma \ref{lemma:matroid-union-theorem}). The lemma asserts that, under matroid-rank valuations,  the sum of the maximin shares is upper bounded by the optimal social welfare. 
 
We will need the following constructs to establish the lemma at the required level of generality. Recall that, for agent $i \in [n]$, the valuation function $v_i$ is the rank function of matroid $\mathcal{M}_i =([m], \mathcal{I}_i)$. For $k \in \mathbb{Z}_+$, write $\mathcal{M}_{i \times k}$ to denote the $k$-fold union of $\mathcal{M}_i$, i.e., 
$\mathcal{M}_{i \times k} \coloneqq \left( [m], \{J_1 \cup \ldots \cup J_k : J_t \in \mathcal{I}_i \text{ for all } t \in [k]  \} \right)$.

Hence, $S\subseteq[m]$ is independent with respect to $\mathcal{M}_{i \times k}$ iff $S$ can be partitioned into $k$ subsets all of which are independent in $\mathcal{M}_i$. Write $r_{i \times k}$ to denote the rank function of $\mathcal{M}_{i \times k}$.

Also, for any subset of agents $B \subseteq [n]$, write $\mathcal{M}_B$ to denote the union of matroids $\{ \mathcal{M}_i \}_{i \in B}$, and $r_B$ to denote the rank function of $\mathcal{M}_B$. 
Note that $\mathcal{M}_{[n]} = \umat$ and the rank function $\urank$ is same as $r_{[n]}$. 

\begin{lemma} \label{lma:mms2}
For any subset of agents $B \subseteq [n]$ and subset of goods $S \subseteq [m]$ we have 
\begin{align*}
\sum_{i \in B} \mu_i(|B|,S)  & \leq r_B(S).
\end{align*}
\end{lemma}

Along the lines of equation (\ref{eq:rank-sw}), we have that $r_B(S)$ is equal to the maximum social welfare that one can achieve by partitioning the subset of goods $S$ among agents in $B$.  

Before proving the lemma we state a supporting proposition--its proof appears in Appendix \ref{appendix:proof-of-mms-prop}.

\begin{restatable}{proposition}{PropositionMMS}
\label{prop:mms}
For an agent $i \in [n]$, subset of goods $S \subseteq [m]$, and integer $k \in \mathbb{Z}_+$, the following inequality holds: $\mu_i(k, S) \leq \frac{r_{i \times k}(S)}{k}$.
\end{restatable}

\begin{proof}[of Lemma \ref{lma:mms2}]
Write $k \coloneqq |B|$. The matroid union theorem (Lemma \ref{lemma:matroid-union-theorem}) applied to the $k$-fold union of $\mathcal{M}_i$ (i.e., to $\mathcal{M}_{i \times k}$) gives us the following convolution formula for the rank function $r_{i \times k}$, for all $i \in [n]$:
\begin{align*}
r_{i \times k}(S) = \min_{T \subseteq S} \left( |S \setminus T| + \sum_{j=1}^k v_i(T) \right) = \min_{T \subseteq S} \left( |S \setminus T| + k \cdot v_i(T) \right). \end{align*}

Dividing both sides of the previous equation by $k$ and using the inequality $\mu_i(k,S) \leq \frac{r_{i \times k}(S)}{k}$ (Proposition \ref{prop:mms}), we get 
$\mu_i(k,S)  \leq \min_{T \subseteq S} \left( \frac{|S \setminus T|}{k} + v_i(T) \right)$. 

Therefore, for any subset $\widetilde{T}\subseteq S$ and all agents $i \in B$, the following upper bound holds $\mu_i(k,S) \leq \left( \frac{|S \setminus \widetilde{T}|}{k} + v_i(\widetilde{T}) \right)$. Summing over $i \in B$ gives us 
\begin{align*}
\sum_{i \in B} \mu_i(k,S) \leq \sum_{i \in B} \left( \frac{|S \setminus \widetilde{T}|}{k} + v_i( \widetilde{T}) \right) = k \cdot \frac{|S \setminus \widetilde{T}|}{k} + \sum_{i \in B} v_i(\widetilde{T}) 
\end{align*}

The previous inequality holds for all subsets $\widetilde{T} \subseteq S$. Therefore, 
\begin{align}
\sum_{i \in B} \mu_i(k,S) & \leq \min_{T \subseteq S} \left( |S \setminus T| + \sum_{i \in B} v_i(T) \right) \label{ineq:rank-conv}
\end{align}

The right-hand-side of inequality (\ref{ineq:rank-conv}) is equal to $r_B(S)$. This follows by  applying the matroid union theorem (Lemma \ref{lemma:matroid-union-theorem}) to the rank function $r_B$ of matroid $\mathcal{M}_B$ (which is the union of $\{\mathcal{M}_i \}_{i \in B}$). Hence, the lemma follows:
$\sum_{i \in B} \mu_i(|B|,S) \allowbreak \leq r_B(S)$.
\end{proof}

Lemma~\ref{lemma:path} (stated and proved below) ensures the existence of path $P$ in Line \ref{line:path} of Algorithm \ref{algo:mms}. The proof of Lemma~\ref{lemma:path} invokes the following result, which is used in the matroid union algorithm as well~\cite[Theorem~42.4]{schrijver2003combinatorial}. 

\begin{lemma}[\cite{schrijver2003combinatorial}]
\label{lemma:sw-gap}
Let $ (A_1, \ldots, A_n) \in \mathcal{I}_1 \times \ldots \times \mathcal{I}_n$ be a partial allocation comprised of independent sets, i.e., $\cup_{i=1}^n A_i$ is an independent set in $\umat$. Then, $|\cup_{i=1}^n A_i| < \urank([m])$ iff there exists a path in the exchange graph $\mathcal{G}(\alloc)$ between $\cup_{i=1}^n F_i(A_i)$ and some good $h \notin \cup_{i=1}^n A_i$ .
\end{lemma}

\begin{restatable}{lemma}{LemmaPath}
\label{lemma:path}
Let $\alloc = (A_1, A_2, \dots , A_n) \in \mathcal{I}_1 \times \mathcal{I}_2 \times \dots \times \mathcal{I}_n$ be a social welfare-maximizing (partial) allocation comprised of independent bundles. Then, for any agent $i \in S_<$, there exists a path in the exchange graph $\mathcal{G}(\alloc)$ from $F_i(A_i)$ to $A_j$, for some $j \in S_>$.
\end{restatable}

\begin{proof}
Recall that $F_i(A_i) \coloneqq \{ g \in [m] \setminus A_i : A_i+g \in \mathcal{I}_i\}$ and $\mu_i = \mu_i(n, [m])$. Define set $R \subseteq [m]$ to be the set of vertices (goods) reachable from the set $F_i(A_i)$ in the exchange graph $\mathcal{G}(\alloc)$. Also, write $B \subseteq [n]$ to denote the set of agents at least one of whose goods is in $R$, i.e., $ B \coloneqq \{k \in [n] \ : \ R \cap A_k \neq \emptyset \} \cup \{i\}$. We explicitly include $i$ in the set $B$.

To begin with, note that for all agents $k \in B$ we have $F_k(A_k) \subseteq R$. This is trivially true if $k = i$. Otherwise, since $k \in B$, there exists a good $g \in R \cap A_k$ ($g$ is reachable from $F_i(A_i)$) and, by definition of exchange graph, there exists an edge from $g \in A_k$ to all the goods in $F_k(A_k)$. Therefore, $F_i(A_i)$ is connected to all of $F_k(A_k)$ and we have $\cup_{k \in B} F_k(A_k) \subseteq R$. 

In addition, we have that $R$ cannot contain an unassigned good, $R \cap \left( [m] \setminus \cup_{i=1}^n A_i \right) = \emptyset$. Since otherwise $\alloc$ would not be a social welfare maximizing allocation--this follows from the reverse direction of Lemma \ref{lemma:sw-gap}. Therefore, every good $g \in R$ satisfies $g \in A_k$ for some agent $k$. Using this observation and the definition of $B$, we obtain $ R \subseteq \cup_{k \in B} A_k$.  

We will now show that there exists a path from $F_i(A_i)$ to a good in $A_j$ for some $j \in S_>$. Assume, towards a contradiction, that such a path does not exist. Equivalently, for all $k \in B$ we have $v_k(A_k) \leq \mu_k$. Note that $i \in B$ and $v_i(A_i) < \mu_i$ imply $\sum_{k \in B} v_k(A_k)  < \sum_{k \in B} \mu_k$.

Now, consider  partial allocation $\alloc' = (A_i)_{i \in B}$, i.e., $\alloc'$ is obtained by restricting $\alloc$ to the set of agents $B$. The previous inequality can be written as 
\begin{align*}
\SW(\alloc')  & < \sum_{k \in B} \mu_k  \leq \sum_{k \in B} \mu_k(|B|, [m]) \\
&  \leq r_B([m]) \tag{via Lemma~\ref{lma:mms2}}
\end{align*}
The penultimate inequality follows from the fact that increasing the number of agents reduces the maximin share: $\mu_k = \mu_k(n, [m]) \leq \mu_k(|B|, [m])$. In particular, with $n \geq |B|$, any $n$-partition of $[m]$ can be transformed into a $|B|$-partition without decreasing the minimum value of the bundles. 

Since $\SW(\alloc') < r_B([m])$, instantiating Lemma \ref{lemma:sw-gap} over the union of $\{M_i\}_{i \in B}$ (i.e., over $\mathcal{M}_B$, instead of $\umat$), we get that there exists a path $P$ in $\mathcal{G}(\alloc')$ (and, hence, in $\mathcal{G}(\alloc)$) from $\cup_{k \in B} F_k(A_k) \subseteq R$ to a good $h \notin \cup_{k \in B} A_k$. Recall that $R \subseteq \cup_{k \in B} A_k$ and, hence, $h \notin R$. 

This, however, contradicts the definition of the reachable set $R$: if there is a path from $\cup_{k \in B} F_k(A_k) \subseteq R$ to $h$, then $h$ is reachable as well. Therefore, the lemma follows. 
\end{proof}

We now establish the main result of this section, which shows that Algorithm \ref{algo:mms}, in polynomial time, finds a social welfare-maximizing allocation that also satisfies the maximin share guarantee. 

\begin{theorem}
\label{theorem:mms}
Every fair division instance, with matroid-rank valuations, admits a maximin share allocation, $\alloc$, that also maximizes social welfare. Furthermore, such an allocation $\mathcal{A}$ can be computed in polynomial time.  
\end{theorem}
\begin{proof}
In Algorithm \ref{algo:mms}, we initialize $\alloc \in \mathcal{I}_1 \times \dots \times \mathcal{I}_n$ as a social welfare-maximizing allocation and execute the while loop. Note that the loop terminates only if $S_< = \emptyset$. In such a case, we have $v_i(A_i) \geq \mu_i$ for all $i \in [n]$, i.e., the allocation in hand $\alloc$ is a maximin share allocation. {Assigning the unallocated goods $[m]\setminus \left(\cup_{i=1}^n A_i \right)$ to the first agent in Line \ref{line:return} does not violate the $\MMS$ guarantee and ensures the algorithm returns a complete allocation.} 

Otherwise, if $S_< \neq \emptyset$ there exists an agent $i \in S_<$ which (by definition of $S_<$) satisfies $v_i(A_i) < \mu_i$. In this case, we can apply Lemma \ref{lemma:path} to infer that, in the exchange graph $\mathcal{G}(\alloc)$, there exists a path from $F_i(A_i)$ to $A_j$, for some $j \in S_>$. In particular, let $P$ be the selected shortest path from $g_1 \in F_i(A_i)$ to $g_t \in A_j$ considered in Line \ref{line:path}.

Given that $\alloc \in \mathcal{I}_1 \times \dots \times \mathcal{I}_n$ and $P$ is a shortest path between $F_i(A_i)$ and $A_j$, applying Lemma \ref{lemma:path-augmentation} we get a new collection of independent sets: $A'_i \coloneqq \left( A_i \Delta P \right) + g_1 \in \mathcal{I}_i$, $A'_j \coloneqq A_j - g_t \in \mathcal{I}_j$, and $A'_k \coloneqq A_k \Delta P \in \mathcal{I}_k$ for each $k \in [n] \setminus \{ i,j \}$. 

Since independence across all the bundles is maintained, we get that $v_i(A'_i) = |A'_i| = |A_i| + 1 = v_i(A_i) + 1$ along with $v_j(A'_j) = |A'_j| = |A_j| - 1 = v_j(A_j) - 1$ and $v_k(A'_k) = |A'_k| = |A_k| = v_k(A_k)$ for each $k \in [n] \setminus \{ i,j \}$. Consequently, the social welfare does not change after the path augmentation. That is, in the algorithm, after the path augmentation in Lines \ref{line:augment1} and \ref{line:augment2}, $\alloc$ continues to be a social welfare-maximizing allocation. 

Recall that $j \in S_>$, hence $v_j(A'_j) = v_j(A_j) - 1 \geq \mu_j$. Therefore, each iteration of the while loop decreases the sum $\sum_{k \in S_<} (\mu_k - v_k(A_k))$ by one. As a result, the total number of while-loop iterations is at most $nm$. That is, after at most $nm$ iterations the while loop terminates with $S_< = \emptyset$. Moreover, each computation performed in a while loop iteration, including exchange graph construction and shortest path computation requires polynomial time.

Therefore, the algorithm finds, in polynomial time, a maximin share allocation that also maximizes social welfare. Note that the guaranteed success of the algorithm implies that such an allocation always exists. 
\end{proof}

We conclude this section by showing that, in contrast to Theorem~\ref{theorem:mms}, existence of maximin share allocations is not guaranteed under two immediate generalizations of matroid-rank functions. In particular, we establish a negative result for (i) binary XOS valuations and (ii) weighted-rank valuations.

\noindent 
(i) \emph{Binary XOS valuations.} A set function $v : 2^{[m]} \mapsto \mathbb{R}_+$ is said to be binary XOS iff, there exists a family of subsets $\mathcal{F} \subseteq 2^{[m]}$, such that $v(S) \coloneqq \max_{F \in \mathcal{F}} |S \cap F|$, for all $S \subseteq [m]$. 

Note that one can express the rank function of any matroid $\mathcal{M} = ([m], \mathcal{I})$ as a binary XOS function, by setting $\mathcal{F} = \mathcal{I}$. Indeed, $\mathcal{F}$ can be exponential in size and, unlike $\mathcal{I}$, is not required to satisfy the hereditary and augmentation properties. 

\noindent 
(ii) \emph{Weighted rank valuations.} For a matroid $\mathcal{M} = ([m], \mathcal{I})$ and weight function $w: [m] \mapsto \mathbb{R}_+$ (which associates a nonnegative weight with each good $g \in [m]$), a set function $v: 2^{[m]} \mapsto \mathbb{R}_+$ is said to be a weighted rank function iff $v(S) \coloneqq \allowbreak \max_{T \subseteq S}\{ \allowbreak \sum_{t \in T} w(t) \ : \ T \in \mathcal{I}\}$, for each $S \subseteq [m]$. Note that weighted rank functions are submodular and, if $w(g) = 1$ for all $g \in [m]$, then $v$ is the rank function of $\mathcal{M}$.

The next two theorems are established by identifying fair division instances that do not admit any maximin share allocation. {In both  the cases, the identified instance has only two agents. That is, these examples rule out the existence of a complete allocation that satisfies the $\PMMS$ criterion and maximizes social welfare as well.} 

\begin{theorem} \label{thm:xos}
Maximin share allocations are not guaranteed to exist for instances in which the agents have binary XOS valuations.
\end{theorem}
\begin{proof}
Consider a fair division instance with two agents, $[n] = \{ 1,2 \}$ and four goods, $[m] = \{1,2,3,4\}$. Define collections $\mathcal{F}_1 \coloneqq \{ \{1,2 \}, \{3,4 \} \}$ and $\mathcal{F}_2 \coloneqq \{ \{1,3 \}, \{2,4 \} \}$. The two agents have binary XOS valuations defined by $\mathcal{F}_1$ and $\mathcal{F}_2$, respectively:  $v_1(S) \coloneqq \max_{X \in \mathcal{F}_1} |S \cap X|$ and $v_2(S) \coloneqq \max_{X \in \mathcal{F}_2} |S \cap X|$, for all $S \subseteq [m]$.

Here, $\mu_1 = \mu_2 = 2$, since each agent $i \in \{1,2\}$ can partition the set of goods, $[m]$, into two bundles of value $2$ for $i$.

Also, note that---irrespective of the valuation class---the requirement that the sum of maximin shares, $\sum_{i \in [n]} \mu_i$, is upper bounded by the optimal social welfare is a necessary condition for the existence of $\MMS$ allocations. We will next show that the current instance does not admit any allocation $\alloc$ such that $\SW(\alloc) \geq 4 = \mu_1 + \mu_2$. This violates the above-mentioned necessary condition for the existence of $\MMS$ allocations and establishes the theorem.

To show that for each allocation $\alloc$ we have $\SW(\alloc) < 4$, note that the reverse inequality ($\SW(\alloc) = v_1(A_1) + v_2(A_2) \geq 4$) can hold only if $v_1(A_1) = v_2(A_2) = 2$. This follows from the fact that both the valuation functions are bounded from above by two. Now, both agents achieve a value of two only if $A_1 \in \mathcal{F}_1$ and $A_2 \in \mathcal{F}_2$. However, by construction, for any subsets $X \in \mathcal{F}_1$ and $Y \in \mathcal{F}_2$ we have $X \cap Y \neq \emptyset$, i.e., $A_1$ and $A_2$ are not disjoint and, hence, $\alloc$ is not an allocation.
\end{proof}

\begin{theorem} \label{thm:weight-rank}
Maximin share allocations are not guaranteed to exist for instances in which the agents' valuations are weighted rank functions. 
\end{theorem}
\begin{proof}
Consider a fair division instance with two agents, $[n] = \{ 1,2 \}$, and four goods, $[m] = \{1,2,3,4\}$. Define $\mathcal{C}$ to be the collection of all subsets of $[m]$ of size at most two, $\mathcal{C} \coloneqq \{ X \subset [m] \ : \ |X| \leq 2 \}$ and weight function $w(\cdot)$ as $w(1) = w(2) = 2$ and $w(3) = w(4) = 1$. With $\mathcal{I}_1 \coloneqq \mathcal{C} \setminus \{ \{ 1, 3\}, \{ 2,4 \} \}$ and $\mathcal{I}_2 \coloneqq \mathcal{C} \setminus \{ \{1, 4\}, \{2, 3\} \}$, let $\mathcal{M}_1 = ([m], \mathcal{I}_1)$ and $\mathcal{M}_2 = ([m], \mathcal{I}_2)$ be two matroids of the two agents, respectively. One can verify that $\mathcal{M}_1$ and $\mathcal{M}_2$ are indeed matroids, i.e., they satisfy both the hereditary and the augmentation property. Here, the valuation $v_i$ of each agent $i \in \{1, 2\}$ is a weighted rank function: $v_i(S) \coloneqq \max_{T \subseteq S}\{ \sum_{t \in T} w(t) \ : \ T \in \mathcal{I}_i \}$, for all $S \subseteq [m]$. 

For agent $1$, $v_1(\{1,4\}) = v_1(\{2,3\}) = 3$, hence $\mu_1 = 3$. Similarly, for agent $2$, $v_2(\{1,3\}) = v_2(\{2,4\}) = 3$, which implies $\mu_2 = 3$. Since the sum of the weights of all the goods is equal to six, for an allocation $\alloc=(A_1, A_2)$ to satisfy the maximin share guarantee it has to be the case that $v_1(A_1) = v_2(A_2) = 3$. That is, $A_1 \in \mathcal{B}_1 \coloneqq \{\{1,4\}, \{2,3\} \}$ and $A_2 \in \mathcal{B}_2 \coloneqq \{\{1,3\}, \{2,4\}\}$. However, for every pair of subsets $X \in \mathcal{B}_1$ and $Y \in \mathcal{B}_2$, we have $X \cap Y \neq \emptyset$. Therefore, $A_1$ and $A_2$ could not be disjoint, i.e., a maximin share allocation $\alloc$ does not exist.
\end{proof}

\section{Pairwise Maximin Share Guarantee for Matroid-Rank Valuations}
This section provides a polynomial-time algorithm (Algorithm \ref{algo:pmms}) for finding a partial allocation that maximizes social welfare (across all allocations) and also satisfies the pairwise maximin share guarantee.

\floatname{algorithm}{Algorithm}
\begin{algorithm}[ht]
  \caption{\textsc{AlgPMMS}} \label{algo:pmms}

  \textbf{Input:} Fair division instance  $\langle [m],[n], \{v_i\}_i \rangle$ with value-oracle access to the matroid-rank valuations $v_i$'s. \\
  \textbf{Output:} Social welfare-maximizing and pairwise maximin share partial allocation $\mathcal{A} = (A_1, \ldots, A_n)$

  \begin{algorithmic}[1]
		\STATE Compute a social welfare maximizing (partial) allocation $\alloc = (A_1, \ldots, A_n) \in \mathcal{I}_1 \times \dots \times \mathcal{I}_n$   
		\WHILE{there exist $i,j \in [n]$ s.t.~that $v_i(A_i) < \mu_i(2, A_i \cup A_j)$} 
			\STATE Set $g$ to be a good in $A_j \cap F_i(A_i) = A_j \cap \{ g' \in [m] \setminus A_i : A_i + g' \in \mathcal{I}_i \}$ \label{line:swap-good}
			\STATE Update $A_i \gets A_i + g$ and $A_j \gets A_j - g$	\label{line:transfer}
			\ENDWHILE
		\STATE \textbf{return } $\alloc = (A_1, A_2, \dots , A_n)$
	\end{algorithmic}
\end{algorithm}

Algorithm \ref{algo:pmms} starts by computing a partial allocation $\alloc = (A_1, \allowbreak \ldots, A_n) \in \mathcal{I}_1 \times \dots \times \mathcal{I}_n$ that maximizes social welfare. As mentioned previously, we can find such an allocation (specifically, with independent bundles) via the matroid union algorithm~\cite[Chapter~42.3]{schrijver2003combinatorial}. 
Also, using the method detailed in Appendix \ref{appendix:compute-mms}, we can compute $\mu_i(2, A_i \cup A_j)$, for all agents $i, j \in [n]$, in polynomial time.  

Algorithm \ref{algo:pmms} iteratively updates $\alloc$ by selecting a pair of agents, $i$ and $j$, between whom the $\PMMS$ criterion does not hold and then it transfers a good from $A_j$ to $A_i$. Below, we will show that a good $g \in A_j \cap F_i(A_i)$ (as required in Line \ref{line:swap-good} of Algorithm \ref{algo:pmms}) necessarily exists and transferring it maintains social welfare. We will complete the proof by establishing that at most a polynomial number of such transfers are required to convert $\alloc$ into a $\PMMS$ allocation.

\begin{theorem}
\label{theorem:pmms}
Every fair division instance, with matroid-rank valuations, admits a partial allocation $\alloc$ that satisfies the pairwise maximin share guarantee and also maximizes social welfare (across all allocations). Furthermore, such a partial allocation $\alloc$ can be computed in polynomial time.
\end{theorem}
\begin{proof}
Let $\alloc = (A_1, \ldots , A_n) \in \mathcal{I}_1 \times \ldots \times \mathcal{I}_n$ be a social welfare-maximizing (partial) allocation such that $v_i(A_i) < \mu_i(2, A_i \cup A_j)$, for some agents $i, j \in [n]$. For such an allocation and pair of agents, we will show that there necessarily exists a good $g \in A_j$ such that transferring $g$ to $A_i$ (i.e., executing Line \ref{line:transfer}) (i) maintains social welfare and (ii) strictly decreases $\sum_{k=1}^n v_k(A_k)^2$.

Since $0 \leq \sum_{k=1}^n v_k(A_k)^2 \leq m^2$, at most $m^2$ such transfers (equivalently, executions of Line \ref{line:transfer}) are possible. That is, after at most $m^2$ iterations, the allocation $\alloc$ in hand will not only continue to maximize social welfare, but will also be $\PMMS$.  

Therefore, to establish the theorem it suffices to show that the desired good transfer can be performed if (in a social welfare maximizing allocation) we have $v_i(A_i) < \mu_i(2, A_i \cup A_j)$, for any pair of agents $i, j \in [n]$. 

The definition of $\mu_i(2, A_i \cup A_j)$ implies that there exists a $2$-partition $(B_1, B_2)$ of the set $A_i \cup A_j$ such that $v_i(B_k) \geq \mu_i(2, A_i \cup A_j)$, for each $k \in \{1, 2\}$. Therefore, for each $k \in \{1, 2\}$, we have\footnote{Recall that, for matroid-rank function $v_i$, we have $v_i(S) \leq |S|$  for all $S \subseteq [m]$ and equality holds only if $S \in \mathcal{I}_i$.}  $|B_k| \geq v_i(B_k) \geq \mu_i(2, A_i \cup A_j) \geq v_i(A_i) + 1 = |A_i| + 1$ . Summing over $k \in \{1, 2\}$ gives us $2  |A_i| + 2 \leq |B_1| + |B_2| = |A_i| + |A_j|$; the last equality follows from the fact that $(B_1, B_2)$ is a partition of $A_i \cup A_j$. Simplifying we get $|A_i| + 2 \leq |A_j|$. {Note that the inequality $|A_i| + 2 \leq |A_j|$ implies that the $\PMMS$ criterion must hold for $j$--otherwise, we would obtain the following contradictory bound: $|A_j| +2 \leq |A_i|$.}
 
Next we will show that there exists a good $g \in A_j \cap F_i(A_i)$. Since $v_i(B_1) > v_i(A_i) = |A_i|$, the augmentation property of matroids ensures that there exists a good $g \in B_1 \setminus A_i$ such that $A_i + g \in \mathcal{I}_i$, i.e., $g \in F_i(A_i)$. Note that $B_1 \subseteq A_i \cup A_j$ and, hence, the good $g$ must be contained in $A_j$. Transferring this good $g \in A_j \cap F_i(A_i)$ from $j$ to $i$ maintains social welfare: $v_j(A_j - g) =  |A_j| - 1$ and $v_i(A_i + g) = |A_i| + 1$. Therefore, even after the transfer, $\alloc$ continues to be a social welfare maximizing allocation.  

Also, recall that $|A_i| + 2 \leq |A_j|$ and, hence, $v_i(A_i + g)^2 + v_j(A_j - g)^2 < v_i(A_i)^2 + v_j(A_j)^2$. This inequality ensures that the sum $\sum_{k=1}^n v_k(A_k)^2$ strictly decreases after every good transfer.

Since transferring $g$ satisfies the required properties (i) and (ii), Algorithm~\ref{algo:pmms} necessarily finds, in polynomial time, a partial allocation that maximizes social welfare and is $\PMMS$ as well. Note that the guaranteed success of the algorithm implies that such an allocation always exists. 
\end{proof}

\section{Conclusion and Future Work}
In this work we established the universal existence of maximin share allocations for matroid-rank valuations. Furthermore, we showed that, in this setting, fairness can be achieved in conjunction with economic efficiency and computational tractability. One can extend the list of desiderata by including truthfulness. As mentioned previously, for matroid-rank valuations, the work of Babaioff et al.~\cite{babaioff2020fair} provides a polynomial-time, truthful mechanism that computes $\EFone$ and Pareto efficient allocations. They additionally note that their mechanism does not necessarily find an $\MMS$ allocation~\cite[Proposition~6]{babaioff2020fair}. Hence, determining whether $\MMS$ admits a truthful mechanism---under rank valuations---is an interesting direction of future work. Note that in the special case of binary additive valuations, every $\EFone$ allocation is $\MMS$ as well and, hence, either one of the truthful mechanisms of Babaioff et al.~\cite{babaioff2020fair} or Halpern et al.~\cite{halpern2020fair} (designed for finding $\EFone$ and Pareto efficient allocations with binary additive valuations) suffices.

Additionally, we proved the existence and efficient computability of partial $\PMMS$ allocations which are simultaneously Pareto optimal. Establishing the existence of complete $\PMMS$ allocations under matroid-rank valuations would also be interesting.

Simple examples show that, in the rank context, $\MMS$ does not imply $\alpha$-approximate $\PMMS$, for any $\alpha>0$. In the reverse direction, we note that any social welfare-maximizing and $\PMMS$ allocation is $\frac{1}{2n}$-approximate $\MMS$ as well (Theorem \ref{theorem:pmms-mms} in Appendix \ref{appendix:pmms-mms}). Extending such results and developing a scale of fairness (as in~\cite{Amanatidis2018}) for rank functions is also interesting.

\section*{Acknowledgements}
Siddharth Barman gratefully acknowledges the support of a Ramanujan Fellowship (SERB - {SB/S2/RJN-128/2015}) and a Pratiksha Trust Young Investigator Award.

\bibliographystyle{alpha} 
\bibliography{fairness,ultimate}

\section*{Appendix}
\appendix
\section{Computing Maximin Shares for Matroid-Rank Valuations}
\label{appendix:compute-mms}
This section shows that if the valuation of an agent is a matroid-rank function, then her maximin share can be computed in polynomial time. 

\begin{theorem}
\label{theorem:compute-maximin}
Given any fair division instance $\langle [m], [n], \{v_i\}_{i=1}^n \rangle$ with matroid-rank valuations, the maximin share $\mu_i(n, [m])$ of each agent $i \in [n]$ can be computed in polynomial time. 
\end{theorem}
\begin{proof}
Fix an agent $i \in [n]$ and let $A = \cup_{t \in [n]} A_t$ be a maximum-size independent set in the $n$-fold union of $\mathcal{M}_i$, i.e., in 
\begin{align*}
\mathcal{M}_{i \times n} \coloneqq \left( [m], \{J_1 \cup \ldots \cup J_n : J_t \in \mathcal{I}_i \text{ for all } t \in [n]  \} \right).
\end{align*}
Note that such a set $A$---along with the independent subsets $A_t \in \mathcal{I}_i$, for all $t \in [n]$---can be computed in polynomial time using the matroid union algorithm~\cite[Chapter~42.3]{schrijver2003combinatorial}.

We will keep updating the collection of subsets $(A_1, \ldots, A_n) \in \mathcal{I}_i \times \ldots \times \mathcal{I}_i$ as long as there exists a pair of indices $j, k \in [n]$ such that $v_i(A_j) \leq v_i(A_k) -2$ (equivalently, $|A_j| \leq |A_k| -2$). Note that in such a case (via the augmentation property) there exists a good $g' \in A_k$ such that $A_j + g' \in \mathcal{I}_i$. We update $A_j \leftarrow A_j + g'$ along with $A_k \leftarrow A_k -g'$ and keep all the other subsets unchanged. The update maintains the independence of the constituent subsets, i.e., the following containment continues to hold $(A_1, \ldots, A_n) \in \mathcal{I}_i \times \ldots \times \mathcal{I}_i$. 
In addition, such an update strictly decreases $\sum_{t \in [n]} v_i(A_t)^2$ (since, $v_i(A_j) \leq v_i(A_k) -2$). Given that $0 \leq \sum_{j \in [n]} v_i(A_j)^2 \leq m^2$, the total number of updates is upper bounded by $m^2$. That is, after at most $m^2$ transfers we will necessarily have a collection $(A_1, \ldots, A_n) \in \mathcal{I}_i \times \ldots \times \mathcal{I}_i$ with the property that 
\begin{align}
v_i(A_j) \geq v_i(A_k) -1 \qquad \text{ for all $j, k \in [n]$} \tag{P}
\end{align}
Write $\left(\widehat{A}_1, \ldots, \widehat{A}_n \right) \in \mathcal{I}_i \times \ldots \times \mathcal{I}_i$ to denote such a collection with property (P). Also, note that $\cup_{i=1}^n \widehat{A}_i$ is a maximum-size independent set in $\mathcal{M}_{i \times n}$. 

We will complete the proof by showing that $\mu_i(n, [m]) =  \min_{t \in [n]}  v_i(\widehat{A}_t)$. The definition of maximin shares ensures that $\mu_i(n, [m]) \geq \min_{t \in [n]}  v_i(\widehat{A}_t)$. 
Now, assume, towards a contradiction, that $\mu_i(n, [m]) > \min_{t \in [n]} \ v_i(\widehat{A}_t)$. This strict inequality and property (P) gives us 
\begin{align}
n \ \mu_i(n, [m]) > \sum_{t=1}^n v_i(\widehat{A}_t) = \sum_{t=1}^n |\widehat{A}_t| = |\cup_{t=1}^n \widehat{A}_t| \label{ineq:comp-mms}
\end{align}
The penultimate equality follows from the fact that $\widehat{A}_t \in \mathcal{I}_i$. 

Next, let $(B_1, \ldots, B_n) \in \mathcal{I}_i \times \ldots \times \mathcal{I}_i$ be a collection of independent subsets that induce $\mu_i(n, [m])$, i.e., $|B_t| = v_i(B_t) \geq \mu_i(n, [m])$, for all $t \in [n]$. Summing over $t$ and using inequality (\ref{ineq:comp-mms}) we get $\sum_{t=1}^n |B_t| \geq n \mu_i(n, [m]) > |\cup_{t=1}^n \widehat{A}_t|$. That is, $| \cup_{t=1}^n B_t | > |\cup_{t=1}^n \widehat{A}_t|$. However, this inequality and the independence of $\cup_{t=1}^n B_t$ (in $\mathcal{M}_{i \times n}$) contradict the fact that $\cup_{i=1}^n \widehat{A}_i$ is a maximum-size independent set in $\mathcal{M}_{i \times n}$. This shows that $\mu_i(n, [m]) =  \min_{t \in [n]}  v_i(\widehat{A}_t)$ and completes the proof. 
\end{proof}

\section{Missing Proofs from Sections~\ref{section:notation} and~\ref{section:mms}}
\label{appendix:mms-proofs}

\subsection{Proof of Lemma~\ref{lemma:path-augmentation}}
\label{appendix:proof-of-path-augmentation}

This section provides a proof of Lemma \ref{lemma:path-augmentation} using the notation of the present paper. We begin by defining a relevant construct and stating two known results: Lemma \ref{lemma:perfect-matching1}~\cite[Theorem~39.13]{schrijver2003combinatorial} and Lemma \ref{lemma:perfect-matching2}~\cite[Corollary~39.13a]{schrijver2003combinatorial}. 

For a matroid $\mathcal{M} = ([m], \mathcal{I})$ and an independent set $A \in \mathcal{I}$, denote by $\mathcal{D}_{\mathcal{M}}(A) = ([m], E_A)$ the directed graph whose  vertex set is $[m]$ and edge set is $E_A \coloneqq \left\{ (g,g') \in   A \times ([m] \setminus A)  \ : \  A-g+g' \in \mathcal{I} \right\}$. Note that for a (partial) allocation $\alloc = (A_1, A_2, \ldots ,A_n) \in \mathcal{I}_1 \times \mathcal{I}_2 \times \ldots \times \mathcal{I}_n$, the set of edges in the exchange graph $\mathcal{G}(\alloc)$ is equal to the union of the edge sets of $\mathcal{D}_{\mathcal{M}_1}(A_1), \ldots, \mathcal{D}_{\mathcal{M}_n}(A_n)$.

\begin{lemma}[\cite{schrijver2003combinatorial}]
\label{lemma:perfect-matching1}
Let $\mathcal{M} = ([m], \mathcal{I})$ be a matroid and let $A \in \mathcal{I}$ be an independent set. Let $B \subseteq [m]$ be such that $|B| = |A|$ and the graph $\mathcal{D}_{\mathcal{M}}(A)$ contains a \emph{unique} perfect matching on the set of vertices $A \Delta B$. Then, the set $B$ is independent as well, $B \in \mathcal{I}$.
\end{lemma}

\begin{lemma}[\cite{schrijver2003combinatorial}]
\label{lemma:perfect-matching2}
Let $\mathcal{M} = ([m], \mathcal{I})$ be a matroid with rank function $r$ and let $A \in \mathcal{I}$ be an independent set. Let $B \subseteq [m]$ be such that $|B| = |A|$, $r(A \cup B) = |A|$, and the graph $\mathcal{D}_{\mathcal{M}}(A)$ contains a \emph{unique} perfect matching on the set of vertices $A \Delta B$. Under these conditions, if for an element $s \notin A \cup B$ we have $A+s \in \mathcal{I}$, then $B+s \in \mathcal{I}$. 
\end{lemma}

We now restate and prove Lemma \ref{lemma:path-augmentation}. 

\LemmaPathAugmentation*

\begin{proof}
Fix any agent $k \in [n] \setminus \{i, j \}$ and write $B_k \coloneqq A_k \Delta P$. The definition of path augmentation gives us $|A_k| = |B_k|$. Furthermore, since $P=(g_1, g_2, \ldots, g_t)$ is a shortest path, the set of edges $N \coloneqq \{ (g_\ell, g_{\ell+1}) \ : \ g_\ell \in A_k \}$ form a unique perfect matching on the set of vertices $A_k \Delta B_k = \{ g_\ell, g_{\ell+1} \ : \ g_\ell \in A_k \}$. Note that all the edges in $N$ are present in the graph $\mathcal{D}_{\mathcal{M}_k}(A_k)$ and, hence, applying Lemma \ref{lemma:perfect-matching1} we get $B_k \in \mathcal{I}_k$, for $k \in [n] \setminus \{i, j \}$.

Since $A_j \in \mathcal{I}_j$, the containment $A_j - g_t \in \mathcal{I}_j$ directly follows from the hereditary property of matroids. For the remaining agent $i$, write $B_i \coloneqq A_i \Delta P$. As in the case of other agents $k$, there exists a unique perfect matching on the set of vertices $A_i \Delta B_i$ in $\mathcal{D}_{\mathcal{M}_i}(A_i)$. In addition, the fact that $P$ is a shortest path from $F_i(A_i)$ to $A_j$ implies $F_i(A_i) \cap (B_i \setminus A_i) = \emptyset$; otherwise there would exist a shorter path starting from (a different vertex of) $F_i(A_i)$ to $A_j$. This non-intersection and the definition of $F_i(A_i)$ ensure that, for each $b \in B_i \setminus A_i$, we have $A_i + b \notin \mathcal{I}_i$. That is, $v_i(A_i \cup B_i) = |A_i|$. Finally, note that $A_i + g_1 \in \mathcal{I}_i$ and, hence, via Lemma \ref{lemma:perfect-matching2}, we get that $B_i + g_1 \in \mathcal{I}_i$. 
\end{proof}

\subsection{Proof of Proposition \ref{prop:mms}}
\label{appendix:proof-of-mms-prop}

This section restates and proves Proposition \ref{prop:mms}.

\PropositionMMS*
\begin{proof}
Recall that $r_{i \times k}$ is the rank function of matroid $\mathcal{M}_{i \times k}$ and $r_{i \times k}(S)$ is equal to the maximum possible social welfare that can be obtained by partitioning $S$ among $k$ copies of agent $i$. That is, for any $k$-partition $(K_1, K_2, \ldots, K_k)$ of $S$ we have 
\begin{align}
\sum_{j=1}^k v_i(K_j) & \leq r_{i \times k}(S) \label{ineq:rik}
\end{align}

Now, towards a contradiction, assume that $k \cdot \mu_i(k,S) > r_{i \times k}(S)$. By definition of $\mu_i(k,S)$, there exists a $k$-partition $\mathcal{B} = (B_1, B_2, \ldots, B_k)$ of  $S$ such that $v_i(B_j) \geq \mu_i(k,S)$ for all $j \in [k]$. This implies that  $\sum_{j=1}^k v_i(B_j) \geq k \cdot \mu_i(k,S) > r_{i \times k}(S)$. This, however, contradicts inequality (\ref{ineq:rik}) and completes the proof. 
\end{proof}

\section{$\PMMS$ implies $\EFone$}
\label{appendix:pmms-efone}
In this section we prove that under monotonic submodular valuations---and, hence, specifically under matroid-rank valuations---every $\PMMS$ allocation is in fact $\EFone$. Recall that a set function $v: 2^{[m]} \mapsto \mathbb{R}_+$ is said to be submodular iff it satisfies $v(A + g) - v(A) \geq v(B + g) - v(B)$ for all subsets $A \subseteq B$ and $g \notin B$.

\begin{theorem}
\label{theorem:pmms-efone}
In any fair division instance with monotonic submodular valuations, if $\alloc$ is a $\PMMS$ (partial) allocation, then $\alloc$ is $\EFone$ as well.  
\end{theorem}
\begin{proof}
We establish the contrapositive form of the claim. Assume that (partial) allocation $\alloc$ is not $\EFone$: there exists a pair of agents $i, j \in [n]$ such that, for every good $g \in A_j$, we have $v_i(A_i) < v_i(A_j - g)$. Since agent $i$ envies agent $j$ and the function $v_i$ is monotonic, we know that $v_i(A_i) < v_i(A_j) \leq v_i(A_i \cup A_j)$. The submodularity of $v_i$ implies that $v_i(A_i) < v_i(A_i \cup A_j) \leq v_i(A_i) + \sum_{g \in A_j} \left( v_i(A_i + g) - v_i(A_i) \right)$. Hence, there exists a good $g' \in A_j$ with the property that $v_i(A_i) < v_i(A_i + g')$. Therefore, we have $v_i(A_i) < \min\{v_i(A_i + g'), v_i(A_j - g')\} \leq \mu_i(2, A_i \cup A_j)$; the last inequality follows from the definition of $\mu_i(2, A_i \cup A_j)$. That is, $\alloc$ does not satisfy the pairwise maximin share guarantee and the theorem follows.
\end{proof}


\section{$\EFone$ does not imply $\MMS$ or $\PMMS$}
\label{example:EFone-not-MMS}
This section provides a two-agent instance (with matroid-rank valuations) wherein an $\EFone$ allocation is neither $\MMS$ or $\PMMS$. In particular, the example here  shows that a converse of Theorem \ref{theorem:pmms-efone} does not hold. 

Consider an instance with $[n] = \{1,2\}$ and $[m] = \{1,2,3,4,5,6\}$. Define set family $\mathcal{I}_1 \coloneqq \{X \subseteq [m] \ : \ |X \cap \{1,2\}|  \leq 1 \text{ and } |X \cap \{3,4\}| \leq 1 \}$ along with matroids $\mathcal{M}_1 = \left([m], \mathcal{I}_1\right)$ and $\mathcal{M}_2 = \left( [m], 2^{[m]} \right)$. The valuation function $v_i$ of each agent $i$ is the rank function of $\mathcal{M}_i$. In this instance, the allocation $\alloc = (A_1, A_2)$, with $A_1 = \{5,6\}$ and $A_2 = \{1,2,3,4\}$, is envy-free (and, hence, $\EFone$) but it is not $\PMMS$ (or $\MMS$), since $\mu_1(2, [m]) = 3$.

\section{$\PMMS$ implies $\frac{1}{2n-1}$-$\MMS$}
\label{appendix:pmms-mms}

Here we provide a scale-of-fairness result between $\PMMS$ and $\MMS$. 
\begin{theorem}
\label{theorem:pmms-mms}
In any fair division instance with matroid-rank valuations, if $\alloc$ is a $\PMMS$ and social welfare-maximizing (partial) allocation, then $\alloc$ is $\frac{1}{2n-1}$-approximate $\MMS$ as well.  
\end{theorem}
\begin{proof}
Since the given (partial) allocation $\alloc$ is $\PMMS$, it is $\EFone$ as well (Theorem \ref{theorem:pmms-efone}). Now, towards a contradiction, assume that $\alloc$ is not $\frac{1}{2n-1}$-$\MMS$, i.e., there exists an agent $i$ for which $v_i(A_i) < \frac{1}{2n-1} \ \mu_i$. Hence, $\mu_i \geq (2n-1) \cdot v_i(A_i) + 1 \geq 1$.

The definition of $\mu_i$ implies that there exists an $n$-partition $(B_1,\ldots, B_n)$ of $[m]$ such that, for all $k \in [n]$, we have $v_i(B_k) \geq \mu_i \geq (2n-1)  v_i(A_i) + 1 \geq 1$. 

First we will show that $v_i(A_i) \geq 1$. Since $v_i(B_k) \geq 1$ for all $k \in [n]$, there exists $n$ distinct goods $G \coloneqq \{g_1, \ldots, g_n\}$ with $g_k \in B_k$ and $v_i(g_k) = 1$ for all $k \in [n]$. Now, if $v_i(A_i) = 0$, then the $n$ goods in $G$ must have be allocated among the remaining $n-1$ agents.\footnote{Any good $g \in G$ cannot be unallocated in $\alloc$, since in such a case assigning it to agent $i$ would increase the social welfare of $\alloc$ contradicting the fact that $\alloc$ maximizes social welfare.} Therefore, there is an agent $j \neq i$ with two goods from $G$ (i.e., $|A_j \cap G| \geq 2$) and this would contradict the fact that $\alloc$ is $\EFone$. Hence, $v_i(A_i) \geq 1$.

For each bundle $B_k$ we have $v_i(B_k) \geq (2n-1) v_i(A_i) + 1$. The augmentation property---invoked between $B_k$ and $A_i$---guarantees the existence of an independent (with respect $\mathcal{I}_i$) subset $S \subset B_k \setminus A_i$ with the property that $|S| = v_i(S) = v_i(B_k) - v_i(A_i) > (2n-2)  v_i(A_i) + 1$. Also, each good $g \in S$ must be allocation in $\alloc$, since otherwise including it in $A_i$ would increase $v_i(A_i)$ and, hence, the social welfare (contradicting the welfare optimality of $\alloc$). Therefore, in $\alloc$, the set $S$ must have been assigned among the remaining $n-1$ agents. In particular, there exists an agent $j$ such that $|A_j \cap S| \geq \frac{1}{n-1} \ |S| > \frac{1}{n-1} \left( (2n-2)  v_i(A_i) + 1 \right) = 2 v_i(A_i) + \frac{1}{n-1}$. Since $|A_j \cap S|$ and $v_i(A_i)$ are integers, we have $|A_j \cap S| \geq 2 v_i(A_i) + 1 \geq v_i(A_i) + 2$. The last inequality follows from the bound $v_i(A_i) \geq 1$. 
 
Finally, we note that $v_i(A_j) \geq v_i(A_j \cap S) = |A_j \cap S| \geq v_i(A_i) + 2$; here, the equality is a consequence of the independence of $S$. This contradicts the fact that $\alloc$ is an $\EFone$ allocation and the theorem follows. 
\end{proof}

\end{document}